\documentclass[lnicst]{svmultln}
\usepackage{makeidx}  % allows for indexgeneration
\usepackage{graphicx}
\usepackage{graphics}
\usepackage{empheq}
\usepackage{amssymb}
\usepackage{amsmath}
\usepackage{multirow}

\begin{document}
\mainmatter              % start of the contribution
\title{The Inverse Task of the Reflexive Game Theory: Theoretical  Matters, Practical Applications and Relationship with Other Issues}
\titlerunning{The Inverse Task} 
% abbreviated title (for running head)
\author{Sergey Tarasenko}
\authorrunning{Sergey Tarasenko}   % abbreviated author list (for running head)
\tocauthor{Sergey Tarasenko}

\institute{Kyoto University, Yoshida honmachi, Kyoto 606-8501, Japan\\
\email{infra.core@gmail.com}}%,\\ WWW home page:

\maketitle              % typeset the title of the contribution
% \index{Tarasenko, Sergey} % entries for the author index

\begin{abstract}        % give a summary of your paper
The Reflexive Game Theory (RGT) has been recently proposed by Vladimir Lefebvre to model behavior of individuals in groups. The goal of this study is to introduce the Inverse task. We consider methods of solution together with practical applications. We present a brief overview of the RGT for easy understanding of the problem. We also develop the schematic representation of the RGT inference algorithms to create the basis for soft- and hardware solutions of the RGT tasks. We propose a unified hierarchy of schemas to represent humans and robots. This hierarchy is considered as a unified framework to solve the entire spectrum of the RGT tasks. We conclude by illustrating how this framework can be applied for modeling of mixed groups of humans and robots. All together this provides the exhaustive solution of the Inverse task and clearly illustrates its role and relationships with other issues considered in the RGT.
%                         please supply keywords within your abstract
\end{abstract}

\section{Introduction}
\label{intro}

The Reflexive Game Theory (RGT) has been entirely developed by Lefebvre \cite{lef1,lef2} and is based on the principles of \textit{anti-selfishness or egoism forbiddeness} \cite{lef1,lef2} and human \textit{reflexion processes} \cite{lef3}. Therefore RGT is based on the human-like decision-making processes.  The main goal of the theory is to model behavior of individuals in the groups. It is possible to predict choices, which are likely to be  made by each individual in the group, and influence each individual's decision-making due to make this individual to make a certain choice. In particular, the RGT can be used to predict terrorists' behavior \cite{lef4}. 

In general, the RGT is a simple tool to predict behavoir of invididuals and influence individuals' choices. Therefore it makes possible to control the individuals in the groups by guiding their behavoir (decision-making, choices) by means of the corresponding influences.

On the other hand, now days robots have become an essential part of our life. One of the purposes robots serve to is to substitute human beings in dangerous situations and environments, like defuse a bomb or radioactive zones etc. 

In contrast, human nature shows strong inclinations towards the risky behavior, which can cause not only injuries, but even threaten the human life. The list of these reasons includes a wide range starting from irresponsible kids' behavior to necessity to find solution in a critical situation. In such a situation, a robot should full-fill a function of refraining humans from doing risky actions and perform the risky action itself, if needed. 

However, robots are forbidden and should not physically force people, but must convince people on the mental level to refrain from doing a risky action. This method is more effective rather than a simple physical compulsion, because humans make the decisions (choices) themselves and treat these decisions as their own. Such technique is called a \textit{reflexive control} \cite{lef3}. 

The task of finding appropriate reflexive control is closely related with the Inverse task, when we need to find suitable influence of one subject on another one or on a group of subject on the subject of interest. Therefore, it is needed to develop the framework of how to solve the Inverse task. This is the primary goal of this study. 

However, for better understanding of the gist of the Inverse task and its intrinsic relationships with other issues of the RGT, we introduce the entire spectrum of the tasks, which can be solved by the RGT. This forms the scope of inference algorithms used in the RGT. We present the RGT algorithms in the form of the \textit{schemas of control systems} that can be instantly applied for developement of soft- or/and hardware solutions. We develop a hierarchy of control systems for abstract individual (including human subject) and robotic agent (robot) based on these control schemas. Finally, we illustrate application of the Inverse task together with other RGT inference algorithms to model robot's behavior in the mixed groups of humans and robots. 

\section{Brief Overview of the Reflexive Game Theory (RGT)}
\label{overview}

\subsection{Representation of groups: graphs, polynomials and stratification tree}
\label{repres}

The RGT deals with groups of abstract subjects (individuals, humans, autonomous agents etc). 
Each subject is assigned a unique variable (\textit{subject variable}). Any group of 
subjects is represented in the shape of \textit{fully connected graph}, which is called 
a \textit{relationship graph}. Each vertex of the graph corresponds to a single subject.
Therefore the number of vertices of the graph is in one-to-one correspondence with overall number of subjects in the groups. Each vertex is named after the corresponding subject variable.

The RGT uses the set theory and the Boolean algebra as the basis for calculus. Therefore the values of subject variables are elements of Boolean algebra. 

All the subjects in the group can have either alliance or conflict relationship. The relationships are identified as a result of
group macroanalysis. It is suggested that the installed relationships can be changed. The relationships are illustrated with graph ribs. The solid-line ribs correspond to alliance, while dashed ones are considered as conflict. For mathematical analysis alliance is considered to be conjunction (multiplication) operation ($\cdot$),  and conflict is defined as disjunction (summation) operation (+). 

The graph presented in Fig.~\ref{fig:fig1}a or any graph containing any sub-graph isomorphic to this graph are not decomposable. In this case, the subjects are excluded from the group one by one, until the graph becomes decomposable.  The exclusion is done according to the importance of the other subjects for a particular one \cite{lef1,lef2}. Any other fully connected graphs are decomposable. Any decomposable graph can be presented in an analytical form of a corresponding \textit{polynomial}. Any relationship graph of three subjects is decomposable (see \cite{lef1,lef2}). 

Consider three subjects $a, b$ and $c$. Let subject $a$ is in alliance with other subjects, while subjects $b$ and $c$ are in conflict (Fig.~\ref{fig:fig1}b). The polynomial corresponding to this graph is $a(b+c)$.
\begin{figure}
\centering
\includegraphics[height=2cm]{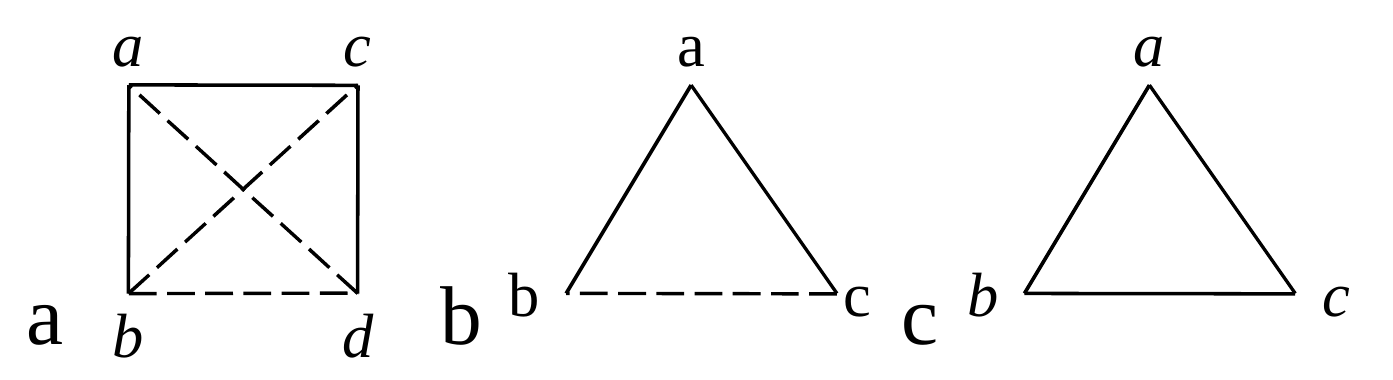}
\caption{The relationship graphs.}
\label{fig:fig1}
\end{figure}
\begin{figure}
\centering
\includegraphics[height=2cm]{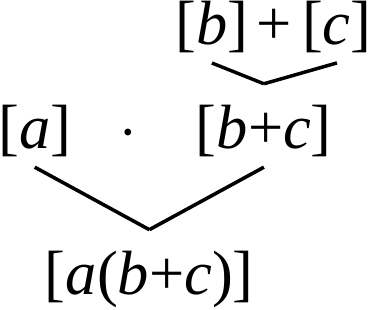}
\caption{Polynomial Stratification Tree. Polynomials $[a], [b]$ and $[c]$ 
are elementary polynomials.}
\label{fig:fig2}
\end{figure}

Regarding a certain relationship, the polynomial can be stratified (decomposed) into \textit{sub-polynomials} 
\cite{lef1, lef2}. Each sub-polynomial belongs to a particular level of stratification. If the stratification regarding alliance was first built, then the stratification regarding the conflict is implemented on the next step. The stratification procedure finalizes, when the \textit{elementary polynomials}, containing a single variable, are obtained after a certain stratification step. 

The result of stratification is the \textit{Polynomial Stratification Tree (PST)}. It has been proved that each non-elementary polynomial can be stratified in an unique way, i.e., each non-elementary polynomial has only one corresponding PST (see \cite{bat:lef} considering one-to-one correspondence between graphs and polynomials). Each higher level of the tree contains polynomials simpler than the ones on the lower level. For the purpose of stratification the polynomials are written in square brackets. The PST for $a(b+c)$ polynomial is presented in Fig.\ref{fig:fig2}.

Next, we omit the branches of the PST and from each non-elementary polynomial write in top right corner its sub-polynomials. The resulting tree-like structure is called a \textit{diagonal form}\cite{lef1, lef2, lef5, lef6}. Consider the diagonal form corresponding to the PST in Fig.~\ref{fig:fig2}:
\[\begin{array}{*{20}{c}}  {} & {} & {[b] + [c]}  \\    {} & {[a][b + c]} & {}  \\    {[a(b + c)]} & {} & {\;\;\;\;\;\;\;\;\;\;\;\;\; .}  \\ \end{array}\]

Hereafter, the diagonal form is considered as a function defined on the set of all subsets of the \textit{universal set}. The universal set contains the \textit{elementary actions}. For example, these actions are actions $\alpha$ and $\beta$. By definition, the \textit{Boolean algebra} of the universal set includes four elements: $1 = \{\alpha , \beta \}$, $ \{\alpha\}$, $\{\beta\}$ and the empty set 0 = $\{\} = \O$. These elements are all the possible subsets of universal set and considered as alternatives that each subject can choose. The alternative $0=\{\}$ is interpreted as an inactive or idle state. In general, Boolean algebra consists of $2^n$ alternatives, if universal set contains $n$ actions. 

Accroding to definition given by Lefebvre \cite{lef5}, we present here exponential operation defined by formula 
\begin{equation}
\label{expfrom}
P^W = P + \overline{W} \ , 
\end{equation}

where $\overline{W}$  stands for negation of $W$ \cite{lef1,lef2,lef4}. 

This exponential operation is used to fold the diagonal form. During the folding, round and square brackets are considered to be interchangeable. The following equalities are also considered to be true: $x + \overline{x}  = 1, x + 0 = x$ and $x + 1 = 1$. Next we implement folding of diagonal form of polynomial $a(b+c)$:
\[\begin{array}{*{20}{c}}
   {} & {} & {[b] + [c]} & {} & {} & {} & {}  \\
   {} & {[a][b + c]} & {} & {} & {} & {[a]([b + c] + \overline{[b] + [c]} )} & {}  \\
   {[a(b + c)]} & {} & {} &  =  & {[a(b + c)]} & {} & { = a(b + c) + \overline a \; .}  \\
\end{array}\]

It is considered that the levels of the PST represent different processing levels of natural or artificial cognitive system. Each level is considered as an images. The root of the tree is the input into the cognitive system and, therefore can be considered as the image of the  world (environment including self and others), perceived by the subject. 

As it follows from the PST, there is a hierarchy of images, corresponding to a particular cognitive level. During processing along this hierarchy in the \textit{bottom-up} manner, the image on the lower level undergoes an extensive process of simplification by the means of decomposition into simpler parts on the higher level. These parts are considered to be the images of the image on the previous level. Therefore, the images on the second level are different representions of the original image of the world. This procedure repeats until we obtain elementary part (elementary polynomials) \cite{lef1, lef2}.

On the other hand, the PST folding procedure can be referred as  \textit{top-down} intergration process of simpler images from the higher levels. 

Therefore, the stratification procedure of original polynomial together with the folding procedure of the diagonal form illustrate the interplay of \textit{bottom-up} and \textit{top-down} information processes, which are widely imployed in biological \cite{kobtan, korner1, luc:mal, korner2} and artificial \cite{fuku, poggio, serre3} information processing systems. The idea of hierarchical structure is highly coherent with hierarchical organization of majority of natural (inanimate objects) and biological (living creatures) entities.  Furthermore, it has been shown that hierarchical structure is intrinsic for the relationships in societies of insects \cite{hienze},  animals \cite{chase1, chase2, buston} and human beings.

Therefore hierarchical representation of the groups in the form of PST correspond to extraction of the hierarchical structure of the given group, while fusion of the PST and its diagonal form with diagonal form folding procedure closely resembles the way of information processing within a single independent congnitive system as discussed above. Thus, RGT imploys the fundamental principles of hierarchical organization on both group (reflects structure of the groups) and individual (illustrates information processing within independent cognitive system of a single unit) levels. This makes RGT universal tools that mildly bridges the gap between representation and analysis.
 
\subsection{The Decision Equation: definition and solution}
\label{deceq}

The goal of each subject in a group is to choose an alternative from the set of alternatives under consideration. 
To obtain choice of each subject, we consider the \textit{decision equations}, which contain subject variable in the left-hand side and the result of diagonal form folding in the right-hand side: 
\begin{eqnarray*}
a = (b+c)a + \overline{a}  \\
b = (b+c)a + \overline{a}  \\
c = (b+c)a + \overline{a}
\end{eqnarray*}
To find solution of the decision equations, we consider the following equation:
\begin{equation}
\label{canfrom}
x = Ax + B\overline{x}  \;  ,
\end{equation}
where  $x$  is the subject variable, and  $A$ and $B$ are some sets.  Eq.(\ref{canfrom}) represents \textit{the canonical form of decision equation}. This equation has solution if and only if the set $B$ is contained in set $A$: $A \supseteq B$. If this requirement is satisfied, then eq.(\ref{canfrom}) has at least one solution from the interval $A \supseteq x \supseteq B$ \cite{lef4}. Otherwise, the decision equation has no solution, and it is considered that subject cannot make a decision. In such situation, the subject is in frustration state.

Therefore, to find solutions of decision equation, one should first transform it into the \textit{canonical form}. Out of three presented equations only the decision equation for subject $a$ is in the canonical form, while other two should be transformed.
We consider explicit transformation only of decision equation for subject $b$ \cite{taras}:\\  \\
$a(b+c)+  \overline{a} =  ab+ac+ \overline{a} = ab + (ac+ \overline{a})b+(ac+\overline{a}) \overline{b} =  (a+  \overline{a} +ac)b+(ac+  \overline{a}) \overline{b} = (1+ac)b+(ac+ \overline{a}) \overline{b}  =  b+(ac +  \overline{a}) \overline{b} = b+(ac +  \overline{a}c +\overline{a}) \overline{b} = b+(c +  \overline{a}) \overline{b}$. \\

Therefore, 
\begin{equation}
\label{canonB}
b = b+(c +  \overline{a}) \overline{b}. 
\end{equation}

The transformation of equation for subject $c$ be can be easily derived by analogy: $c = c + (b +  \overline{a}) \overline{c}$.

Next we consider two tasks, which can be formulated regarding the decision equation in the canonical form and provide methods to solve each task.

\subsection{The Forward Task}
\label{forwardT}

The variable in the left-hand side of the decision equation in canonical form is the variable of the equation, while other variables are considered as influences on the subject from the other subjects. The  \textit{Forward task} is formulated as a task to find the possible choices of a subject of interest, when the influences on him from other subjects are given.

After transformation of arbitral decision equation into its canonical form, the sets $A$ and $B$ are functions of other subjects' influences. For example, if we consider group of subjects $a$, $b$, $c$, etc. togehter with the abstract representation of decision equation in canonical form for subject $a$, the sets $A$ and $B$ will be the functions of subject variables $b$, $c$, etc. :
\begin{equation}
\label{ftcnf}
a = A(b,c,...)a + B(b,c,...)\overline{a}  \  .
\end{equation}

In the case of only three subjects $a$, $b$ and $c$, $A(b,c,...) = A(b,c)$ and $B(b,c,...) = B(b,c)$.

All the influences are presented in influence matrix (Table \ref{infMat}). The main diagonal of influence matrix contains the subject variables. The rows of the matrix represent influences of the given subject on other subjects, while columns represent the influences of other subjects on the given one. The influence values are used in decision equations.

\begin{table}
\caption{Influence Matrix}
\label{infMat}
\begin{center}
%\begin{tabular}{r@{\quad}rl}
%\begin{tabular}{|r@{\quad}|r@{\quad}|r@{\quad}|r@{\quad}|}
\begin{tabular}{|c|c|c|c|}
\hline
{}&a&b&c\\%[2pt]
\hline
\rule{0pt}{12pt}a&a&$\{\alpha\}$&$\{\beta\}$\\[2pt]
\hline
\rule{0pt}{12pt}b&$\{\beta\}$&b&$\{\beta\}$\\[2pt]
\hline
\rule{0pt}{12pt}c&$\{\beta\}$&$\{\beta\}$&c\\[2pt]
\hline
\end{tabular}
\end{center}
\label{table:tab1}
\end{table}

For subject $a$: $a = (\{\beta\}+\{\beta\}) a + \overline{a}  \Rightarrow  a = \{\beta\}a + \overline{a}$.

For subject $b$: $b = b + (\{\alpha\}\{\beta\}+\overline{\{\alpha\}})\overline{b}    \Rightarrow b = b + 
\{\beta\}\overline{b}$.

For subject $c$: $c = c + (\{\beta\}\{\beta\}+\overline{\{\beta\}})\overline{c}  \Rightarrow 
c = c+(\{\beta\} +\{\alpha\})\overline{c}  \Rightarrow  c = 1$.

Equation for subject $a$ does not have any solutions, since set $A = A(b,c) = \{\beta\}$ is contained in set $B = B(b,c) = 1$: $A \subset B$.  Thus, subject $a$ cannot make any decision. Therefore he is considered to be in frustration state.

Equation for subject $b$ has at least one solution, since $ A = A(b,c) = 1 = \{\alpha, \beta\}\supseteq B = B(b,c) = \{\beta\}$. The solution belongs to the interval $1\supseteq b \supseteq \{\beta\}$. Therefore subject $b$ can choose any alternative from Boolean algebra, which contains alternative $\{\beta\}$. These alternatives are $1 = \{\alpha,\beta\}$ and $\{\beta\}$.

Equation for subject $c$ turns into equality $c = 1$. This is possible only in the case, when $A(b,c) \equiv B(b,c)$. Here $A = B = 1$. 

%!!!!!!!!!!!!!!!!!!!!!!!!!!!
\subsection{The Inverse Task}
\label{inverseT}

In contrast to the Forward task, the \textit{Inverse task} is formulated as a task to find all the simultaneous (or joint) influences of all the subjects together on the subject of interest that result in choice of a particular alternative or subset of alternatives. We call the subject of interest to be a \textit{controlled subject}.

Let subject $a$ be a controlled subject and $a^*$ is a fixed value, representing an  alternative or subset of alternatives, which subjects $b$, $c$, etc. want subject $a$ to choose. We call value $a^*$ to be a \textit{target choice}. By substituting subject variable $a$ with fixed value $a^*$, we obtain the \textit{influence equation}. If we substitute the subject variable $a$ with fixed value $a^*$ in the canonical form of the decision equation (eq. (\ref{ftcnf})), we obtain \textit{the canonical form of the influence equation}: 
\begin{equation}
\label{infG1}
a^* = A(b,c,...)a^* + B(b,c,...)\overline{a^*}  \;  ,
\end{equation}

For only three subjects $a$, $b$ and $c$, $A(b,c,...) = A(b,c)$ and $B(b,c,...) = B(b,c)$.

In contrast to the decision equation, which is equation of a single variable, the influence equation is the equation of multiple variables. However, the number of variables of influence equation is not trivial question. 
In fact, the number of variables in influence equation can be less then $(n-1)$, where $n$ is the total number of subjects in the group. There are groups, in which sets $A$ and $B$ are functions of less than $(n-1)$ variables (see Appendix \ref{appen1}). Therefore the variables that present in influence equation are called \textit{effective variables}.

The Inverse task is by definition\footnote{We need a system of influence equations because solutions of the influence equation $a^* = A(b,c,...)a^* + B(b,c,...)\overline{a^*}$ itself only guaratee that the original decision equation $a = A(b,c,...)a + B(b,c,...)\overline{a}$ turns into true equality, but it is not guaranteed that these solutions are the only ones that turn decision equation into true equality.} 
formalized as to find all the joint solutions of all subjects in the group, except for the controlled one, when the target choice is represented by interval $\chi_1 \supseteq a^* \supseteq \chi_2$, where $\chi_1$ and $\chi_2$ are some sets and $\chi_1 \supset \chi_2$. In such a case, to solve the Inverse task, one should solve the system of influence equations:
\begin{empheq}[left=\empheqlbrace]{align}
A(b,c,...) = \chi_1 \label{sys11G} \\
B(b,c,...) = \chi_2  \label{sys21G} 
\end{empheq}

If the target choice is a single alternative, then $\chi_1 = \chi_2 = a^*$.

\textit{The solutions of the system (\ref{sys11G}-\ref{sys21G}) are considered as reflexive control strategies.}

The solution of the Inverse task in particular is characterized from two points. The first point is whether it is required to find the influence of a particular single subject or joint influences of a group of subjects. The second one is whether the target choice is represented as a single alternative or as an interval of alternatives. 

To illustrate these points,  we introduce a particular group of subjects. Let subjects $a$ and $b$ are in alliance with each other and in conflict with subject $c$. 
%(Fig.\ref{fig:fig1}b.
The polynomial corresponding to this graph is $ab+c$. The diagonal form corresponding to this polynomial and its folding is  
\[\begin{array}{*{20}{c}}
   {} & {} & {[a][b]} & {} & {}  \\
   {} & {[ab]} & {} & { + [c]} & {}  \\
   {[ab + c]} & {} & {} & {} & { = ab + c}  \\
\end{array}\]

Therefore the decision equation for all the subjects in the group is 
\begin{equation}
x = ab + c,
\label{eq1}
\end{equation}
where $x$ can be any subject variable $a$, $b$ or $c$. 

\textit{Influence of a single subject vs joint influences of a group.} First we consider example, when the influence of a single subject is required. Let subject $b$ makes influence $\{\alpha\}$ and  $a^* = \{\alpha\}$. Then we need to find influences of a single subject $c$, which result in solution $a^* = \{\alpha\}$ of decision equation $a = ab+c$. 

The canonical form of this influence equation is  $a^* = (b+c)a^*+c\overline{a^*}$. Since  $a^* = \{\alpha\}$,  $\chi_1= \chi_2=\{\alpha\}$, we obtain a system of equations:
\begin{empheq}[left=\empheqlbrace]{align}
\{\alpha\} + c = \{\alpha\} \label{sys11} \\
c = \{\alpha\}  \label{sys21} 
\end{empheq}
Therefore, the straight forward solution of this system is $c = \{\alpha\}$.

This simple example illustrates the very gist of the \textit{Inverse task} - to find the appropriate influences, which result in target choice.

Next, we consider that influence of subject $b$ is not known. Therefore, we obtain system
\begin{empheq}[left=\empheqlbrace]{align}
b + c = \{\alpha\} \label{sysinf11} \\
c = \{\alpha\}  \label{sysinf12} 
\end{empheq}

In this case, we need to find the values of variable $b$, which together with $c$, result in solution  $a^* = \{\alpha\}$. In other words, we need to find all the pairs $(b,c)$, resulting in solution  $a^* = \{\alpha\}$. These pairs are solutions of the system (\ref{sysinf11}-\ref{sysinf12}). Therefore, we run all the possible values of variable $b$ and check if the first equation of the system (\ref{sysinf11}-\ref{sysinf12}) turns into true equality: \\
$b = 1: 1 + \{\alpha\} = 1 \Rightarrow 1 \neq \{\alpha\}$; \\
$b = \{\alpha\}: \{\alpha\} + \{\alpha\} = \{\alpha\} \Rightarrow \{\alpha\} = \{\alpha\}$; \\
$b = \{\beta\}: \{\beta\} + \{\alpha\} = 1 \Rightarrow 1 \neq \{\alpha\}$; \\
$b = 0: 0 + \{\alpha\} = \{\alpha\} \Rightarrow \{\alpha\} = \{\alpha\}$. 

Therefore, out of four possible values of variable $b$, only two values $\{\alpha\}$ and $0$ are appropriate. Thus, we obtain two pairs $(b,c)$: $(\{\alpha\},\{\alpha\})$ and $(\{\alpha\},0)$.

\textit{A single target alternative vs interval of alternatives.} In the previous examples we considered a target choice to be only a single alternative. Here we illustrate the case, when a target choice is an interval. Let $b = \{\beta\}$, and $1 \supseteq a^* \supseteq \{\alpha\}$. To find corresponding influences of subject $c$, we solve the system of equations:
\begin{empheq}[left=\empheqlbrace]{align}
\{\beta\} + c = 1 \label{sysinf21} \\
c = \{\alpha\}  \label{sysinf22} 
\end{empheq}

Again, we instantly obtain the solution of this system: $c = \{\alpha\}$.

In this section, we have formulated the Inverse task in general and considered its particular formalization depending on the number of influences and what is the target choice. However, we do not have a method to solve arbitral influence equation. Therefore, we solve this problem in the next section.

\section{How to Solve an Arbitral Influence Equation}
\label{solveInvT}

As an introduction for this section, we consider the fundamental proposition, which will be the conner stone to solve the influence equations.
\begin{proposition}
\label{lem1}
Let P and Q be some abstract sets. Then $P\overline{Q} + \overline{P}Q = 0 \Leftrightarrow P = Q$. 
\end{proposition}
\begin{proof}
\textit{Necessity.} Let $P\overline{Q} + \overline{P}Q = 0$, then 
\begin{eqnarray*}
P\overline{Q} + \overline{P}Q = 0 \Rightarrow P\overline{Q} + \overline{P}Q + P = P \Rightarrow P + \overline{P}Q = P \Rightarrow \\
P(Q + \overline{Q}) + \overline{P}Q = Q + P\overline{Q} + \overline{P}Q = P \Rightarrow Q = P. \\
\end{eqnarray*}
Therefore if $P\overline{Q} + \overline{P}Q = 0$, then $P = Q$.

\textit{Sufficiency.} Let  $P = Q$, then $P\overline{P} + \overline{P}P = 0$. $\Box$ 
\end{proof}

Now let us consider the new type of equation:
\begin{equation}
A_1x + B_1\overline{x} = 0
\label{equation:eq3}
\end{equation}
This equation has solution if and only if $\overline{A_1} \supseteq x \supseteq B_1$. 

\subsection{Solving Influence Equations}
\label{infEqsInt}

There are three operations defined on the Boolean algebra. They are conjunction ($\cdot$ or multiplication), disjunction (+ or summation) and negation ($\overline{x}$, where $x$ is subject variable). The negation operation is unary operation, while other two operations are binary. Using combination of these three operations, we can compose any influence equation. Since, it is obvious how to solve the equation including only unary operation, we discuss how to solve influence equations including a single binary operation.

For this perpose, we consider two abstract subject variables $x_1$ and $x_2$ and abstract alternative $\chi$.

%!!!!!!!!!!!!!!!!!!!!!!!!!!!!!!!!!!!!!!!!!!!!!!!!!!!!!!!!!!

\begin{lemma}
\label{lem}
The solution of equation 
\begin{equation}
\label{disE}
x_1 + x_2 = \chi
\end{equation}
regarding variable $x_i$, where $i = 1,2$, is given by the interval 
$\chi \supseteq x_i \supseteq (\overline{\chi}x_j + \overline{x_j}\chi)$, where $j = 1,2; j \neq i$.
\label{lemma1}
\end{lemma}
\begin{proof}
According to Proposition 1, $P = x_1 + x_2$, $Q = \chi$, $\overline{P} = \overline{x_1 + x_2} =\overline{x_1}$ $\overline{x_2}$ and $\overline{Q} = \overline{\chi}$.

Therefore, $P\overline{Q}  + \overline{P}Q = (x_1 + x_2)\overline{\chi} +\overline{x_1}$ $\overline{x_2}\chi  = 
x_1\overline{\chi} + x_2\overline{\chi} +\overline{x_1}$ $\overline{x_2}\chi$.
Consequently, we obtain eq.(\ref{disEq}):
\begin{equation}
\label{disEq}
x_1\overline{\chi} + x_2\overline{\chi} +\overline{x_2} \chi \overline{x_1} = 0
\end{equation}

We solve eq.(\ref{disEq}) regarding variable $x_1$. First, we transform eq.(\ref{disEq}) into canonical form:
\begin{equation}
\label{disEqcn}
\overline{\chi}x_1 + (\overline{\chi}x_2 + \chi \overline{x_2}) \overline{x_1} = 0
\end{equation}

Therefore, the solution of eq.(\ref{disEqcn}) is given by the interval 
\begin{equation}
\label{disEqcnInt}
\chi \supseteq x_1 \supseteq (\overline{\chi}x_2 + \overline{x_2}\chi) .
\end{equation}

Since variables $x_1$ and $x_2$ are interchangable and it is possible to solve eq.(\ref{disEq}) regarding variable $x_2$ as well, the general form of solution of eq.(\ref{disE}) is the interval 
\begin{equation}
\label{disEqcnIntG}
\chi \supseteq x_i \supseteq (\overline{\chi}x_j + \overline{x_j}\chi) .
\end{equation}

where $i = 1,2$ and $j = 1,2; j \neq i. \Box$
\end{proof} 

%!!!!!!!!!!!!!!!!!!!!!!!!!!!!!!!!!!!!!!!!11

\begin{lemma}
\label{lem3}
The solution of equation 
\begin{equation}
\label{conE}
x_1 x_2 = \chi
\end{equation}
regarding variable $x_i$, where $i = 1,2$, is given by the interval 
$(\chi x_j + \overline{\chi} \ \overline{x_j}) \supseteq x_i \supseteq \chi $, where $j = 1,2; j \neq i$.
\label{lemma2}
\end{lemma}
\begin{proof}
According to Proposition 1, $P = x_1x_2$, $Q = \chi$, $\overline{P} = \overline{x_1x_2} =\overline{x_1}+\overline{x_2}$ and $\overline{Q} = \overline{\chi}$.

Therefore, $P\overline{Q}  + \overline{P}Q = (x_1x_2)\overline{\chi} +(\overline{x_1}$ + $\overline{x_2})\chi  = 
x_2\overline{\chi}x_1 + \overline{x_1}\chi +\overline{x_2}\chi$.

Thus, we obtain eq.(\ref{conEq}):
\begin{equation}
\label{conEq}
x_2\overline{\chi}x_1 + \overline{x_1}\chi +\overline{x_2}\chi = 0
\end{equation}

We solve eq.(\ref{conEq}) regarding variable $x_1$. First, we transform eq.(\ref{conEq}) into canonical form:
\begin{equation}
\label{conEqcn}
(\overline{\chi}x_2 + \chi \overline{x_2})x_1 + \chi \overline{x_1} = 0
\end{equation}

Since $\overline{\overline{\chi} x_2 + \chi \overline{x_2} } = \chi x_2 + \overline{\chi} \ \overline{x_2}$, the solution of 
eq.(\ref{conEqcn}) is given by the interval 
\begin{equation}
\label{conEqcnInt}
(\chi x_2 + \overline{\chi} \ \overline{x_2}) \supseteq x_1 \supseteq \chi .
\end{equation}

Since variables $x_1$ and $x_2$ are interchangable and it is possible to solve eq.(\ref{conEq}) regarding variable $x_2$ as well, the general form of solution of eq.(\ref{conE}) is the interval 
\begin{equation}
\label{conEqcnIntG}
(\chi x_j + \overline{\chi} \ \overline{x_j}) \supseteq x_i \supseteq \chi  .
\end{equation}
where $i = 1,2$ and $j = 1,2; j \neq i. \Box$
\end{proof}

%!!!!!!!!!!!!!!!!!!!!!!!!!!!!!!!!!!!!!!!!!!!!!!!!!!!!!!!!!!!!!!!!!!!1

Since one bound of the solution intervals for eqs.(\ref{disE}) and (\ref{conE}) are functions of the second variable, we need to run all the possible values of the second variable in order to obtain all possible solutions of these equations in the form of pairs $(x_1, x_2)$.

Next we consider several examples, illustrating application of Lemmas \ref{lemma1} and \ref{lemma2}.

\textit{Example 1}. For illustration, we solve equation $a^* = ba^*+c$. Consider $\chi = a^*$, $x_1 = ba^* $ and $x_2  = c$, we obtain the solution interval for variable $x_2 = c$: $\chi \supseteq c \supseteq (\chi\overline{\chi b} + \overline{\chi} \ \chi b)$. After simplfication, we get interval (\ref{inter1}):
\begin{equation}
\label{inter1}
\chi \supseteq c \supseteq \chi \overline{b}
\end{equation} 

Next we consider examples with particular alternatives. Let it be alternative $\{\alpha\}: \chi = \{\alpha\}$. The solution interval is then $\{\alpha\} \supseteq c \supseteq \{\alpha\}\overline{b}$.  Since the lower bound of this interval is a function of variable $b$, to find all solutions of equation $a^* = ba^*+c$, we calculate value of expression $\{\alpha\}\overline{b}$ for all possible values of variable $b$ (Table \ref{solns1}).

To reesure that solutions are correct, we check that decision equation $a = ba+c$ turns into true equality for the obained pairs $(b,c)$: \\

$(\{\alpha\}, \{\alpha\})$: $\{\alpha\}\{\alpha\} + \{\alpha\} = \{\alpha\} \Rightarrow \{\alpha\} = \{\alpha\}$ is true;

$(\{\alpha\}, 0)$: $\{\alpha\}\{\alpha\} + 0 = \{\alpha\} \Rightarrow \{\alpha\} = \{\alpha\}$ is true;

$(\{\beta\}, \{\alpha\})$: $\{\alpha\}\{\beta\} + \{\alpha\} = \{\alpha\} \Rightarrow \{\alpha\} = \{\alpha\}$ is true;

$(1, \{\alpha\})$: $\{\alpha\}1 + \{\alpha\} = \{\alpha\} \Rightarrow \{\alpha\} = \{\alpha\}$ is true; 

$(1, 0)$: $\{\alpha\}1 + 0 = \{\alpha\} \Rightarrow \{\alpha\} = \{\alpha\}$ is true;

$(0, \{\alpha\})$: $\{\alpha\}0 + \{\alpha\} = \{\alpha\} \Rightarrow \{\alpha\} = \{\alpha\}$ is true. \\

So far, we have illustrated how to solve the influence equation. We as well showed that the pairs $(b,c)$ obtained by solving equation $a^* = ba^*+c$ in accordance with Proposition 1 and Lemmas 1 and 2 are indeed solutions of this equation.

\begin{table}
\caption{Solutions of the influence equation $a^* = ba^*+c$}
\begin{center}
\begin{tabular}{c c c c c}
%\hline
\hline
\rule{0pt}{12pt}{Values of $b$}&{$\{\alpha\}$}&{$\{\beta\}$}&{1}& {0} \\
\hline
\rule{0pt}{12pt} \multirow{2}{*}{Pairs $(b,c)$} & {$(\{\alpha\},\{\alpha\})$} & {$(\{\beta\},\{\alpha\})$} & {$(1, \{\alpha\})$} & {$(0, \{\alpha\})$} \\[2pt]
\rule{0pt}{12pt} &{$(\{\alpha\},0)$} & {} & {$(1, 0)$} & {} \\[2pt]
\hline
\label{solns1}
\end{tabular}
\end{center}
\label{solns1}
\end{table}

\textit{Example 2}. We consider influence equation for subject $b$ obtained from eq.(\ref{canonB}).
\begin{equation}
\label{ex2eq}
(c+\overline{a})\overline{\chi} + \chi = \chi
\end{equation}

First, we transform the left-hand side of eq.(\ref{ex2eq}):\\

$(c+\overline{a})\overline{\chi} + \chi = c \overline{\chi}+\overline{a} \overline{\chi} + \chi =   c \overline{\chi}+\overline{a} \overline{\chi} + (c+\overline{a}+1)\chi = c + \overline{a} + \chi$. \\

Therefore, eq.(\ref{ex2eq}) can be rewritten as follows:
\begin{equation}
\label{ex2eq1}
c + \overline{a} + \chi = \chi
\end{equation}

Considering, $x_1 = c$ and $x_2 = \overline{a} + \chi$, we instantly obtain the solution interval of eq.(\ref{ex2eq1}): $ \chi \supseteq c \supseteq (\overline{\chi}(\overline{a} + \chi) + \chi (\overline{\overline{a} + \chi})) \Rightarrow \chi \supseteq c \supseteq (\overline{\chi} \ \overline{a} + \chi \overline{\chi}a)$.

Finally,
\begin{equation}
\label{inter2}
\chi \supseteq c \supseteq \overline{\chi} \ \overline{a}
\end{equation}

\textit{Example 3}. Next, we consider influence equation 
\begin{equation}
\label{ex3eq}
ab+\chi = \chi
\end{equation}

Considering, $x_1 = ab$ and $x_2 = \chi$, we instantly obtain the solution interval  $\chi \supseteq ab \supseteq (\chi \overline{\chi} + \overline{\chi}\chi )$ or
\begin{equation}
\label{inter3}
\chi \supseteq ab \supseteq 0
\end{equation}

Therefore, in order to find all solutions of eq.(\ref{ex3eq}), we need to solve the equations
\begin{equation}
 ab = y 
 \end{equation}
where $y$ is any sub-set of set $\chi$ ($y \supseteq \chi$).

Each equation can be solved according to Lemma \ref{lemma2}.

\textit{Example 4}. As a final example, we again consider influence equation $a^* = (b+c)a^*+c\overline{a^*}$ and show how application of Lemma \ref{lemma1} essentially simplifies its solution. We get the system of influence equations:
\begin{empheq}[left=\empheqlbrace]{align}
b + c = \{\alpha\} \ ; \label{sys2a} \\
c = \{\alpha\} \ .  \label{sys2b} 
\end{empheq}

From this system we obtain a single equation:
\begin{equation}
\label{deqs}
b + \{\alpha\} = \{\alpha\} \ .
\end{equation}

According to Lemma \ref{lemma1}, we instantly obtain the solution interval of eq.(\ref{deqs}):
\begin{equation}
\label{inter4}
\{\alpha\} \supseteq b \supseteq 0 \ .
\end{equation}

Thus, eq.(\ref{deqs}) has two solutions: $b = \{\alpha\}$ and $b = 0$. Therefore the solution of system (\ref{sys2a}-\ref{sys2b}) consists of two pairs $(\{\alpha\}, \{\alpha\})$ and $(0,\{\alpha\})$.

To conclude this section, we provide its brief summary. We have shown how to solve the Inverse task by means of influence equations. We have proved two fundamental lemmas, which allow to solve any influence equation regardless of the number of variables. Finally, we have illustrated several examples of how apply these lemmas.

%!!!!!!!!!!!!!!!!!!!!!!!!!!!!!!!!!!!!!!!!!!!
\subsection{Analysis of Extreme Cases 1: Frustration}
\label{frust}

In this section we analyze the situation, when subject can appear in frustration state, from the point of view of the inverse task. Let us consider the polynomial $a(b+c)$ discussed in the section \ref{repres}. The decision equation that corresponds to this polynomial is $x = (b+c)a + \overline{a}$, where $x$ can be any subject variable.

Next we try to find all the pairs $(b,c)$ such that result in selection of a particular alternative by subject $a$.

The decision equation for subject $a$ is $a = (b+c)a + \overline{a}$. The solution interval of this decision equation is $b+c \supseteq a \supseteq 1$. We need to check which alternative subject $a$ can be convinced to choose. To do this, we consider the system of equation for each alternative.

Alternative $\{\alpha\}$:
\begin{empheq}[left=\empheqlbrace]{align}
b + c = \{\alpha\} \label{sys11} \\
1 = \{\alpha\}  \label{sys21} 
\end{empheq}

Alternative $\{\beta\}$:
\begin{empheq}[left=\empheqlbrace]{align}
b + c = \{\beta\} \label{sys12} \\
1 = \{\beta\}  \label{sys22} 
\end{empheq}

Alternative $0=\{\}$:
\begin{empheq}[left=\empheqlbrace]{align}
b + c = 0 \label{sys13} \\
1 = 0  \label{sys23} 
\end{empheq}

In these systems the second equation is incorrect equality. Therefore these systems have no solution.

Alternative $1=\{\alpha, \beta\}$:
\begin{empheq}[left=\empheqlbrace]{align}
b + c = 1 \label{sys14} \\
1 = 1  \label{sys24} 
\end{empheq}
The second equation is correct equality. Therefore this system has solution. 

Thus, out of four possible alternatives, subject $a$ actually can choose only alternative $1=\{\alpha, \beta\}$.
To find solutions, resulting in selection of the alternative $1=\{\alpha, \beta\}$,  we need to solve only eq.(\ref{sys14}), since eq.(\ref{sys24}) turns into the true equality. 

According to Lemma \ref{lemma1}, we instantly obtain the solution interval for eq.(\ref{sys14}): 
\begin{equation}
1 \supseteq b \supseteq \overline{c}
\label{inter51}
\end{equation}

We calculate the pairs $(b,c)$ for all possible values of variable $c$ (Table \ref{tab3}). 

\begin{table}
\caption{Solutions of the influence equation $b + c = 1$}
\label{tab3}
\begin{center}
\begin{tabular}{c c c c c}
%\hline
\hline
\rule{0pt}{12pt}{Values of $c$}&{$\{\alpha\}$}&{$\{\beta\}$}&{1}& {0} \\
\hline
\rule{0pt}{12pt} \multirow{4}{*}{Pairs $(b,c)$} & {$(\{\beta\},\{\alpha\})$} & {$(\{\alpha\},\{\beta\})$} & {$(0,1)$} & {$(1,0)$} \\[2pt]
\rule{0pt}{12pt} &{$(1,\{\alpha\})$} & {$(1,\{\beta\})$} & {$(\{\alpha\},1)$} & {} \\[2pt]
\rule{0pt}{12pt} &{} & {} & {$(\{\beta\},1)$} & {} \\[2pt]
\rule{0pt}{12pt} &{} & {} & {$(1,1)$} & {} \\[2pt]
\hline
\label{tab3}
\end{tabular}
\end{center}
%\label{table:tab3}
\end{table}

Therefore, the influence analysis of the decision equation $a = (b+c)a + \overline{a}$ shows that the only alternative that subject $a$ can choose is alternative $1=\{\alpha,\beta\}$. The influence analysis provides us with the set (exhaustive list) of pairs $(b,c)$ of joint influences resulting in selection of alternative $1=\{\alpha, \beta\}$. Therefore, if the pair of influences does not match any pair from this list, the decision equation has no solution and this results in frustration state. 

Summarizing, this section we note that in general there are two sets. The set $\mathbb D$ contains alternatives that a controlled subject can choose. The set $\mathbb U$ is the set of altertanives of the target choice. Therefore, the need to put subject $a$ into frustration state emerges, if the target choice of a controlled subject cannot be made by this subject. In other words, we need to put a subject into frustration state, if $\mathbb D \cap \mathbb U = \O$.

%!!!!!!!!!!!!!!!!!!!!!!!!!!!!!!!!!!!!!!!!!!!!!!
\subsection{Analysis of Extreme Cases 2: What to do with Super-Active Groups}
\label{supac}

Among all the possible groups, there are groups, in which subjects will always choose only the alternative $1=\{\alpha, \beta\}$ regardless of the influence of other subjects. Such groups are called \textit{super-active groups}.

Next we consider one special case of super active groups - the $homogenous$ groups. The group is called $homogenous$, if all the subjects in the group are connected with the same relationship.

Here we provide proof of the lemma about homogenous groups originally formulated by Lefebvre \cite{lef1,lef2}.

\begin{lemma}
\label{lem4}
Any homogenous group is the super-active group.
\label{lemma3}
\end{lemma}
\begin{proof}
We consider the homogenous groups, where all the subjects are connected with alliance (alliance groups) and conflict (conflict groups) relationship, separately.

Without loss of generallity, we suggest that there are $n$ subjects $a_1, a_2, ..., a_n$.

\textit{Alliance groups}.
The polynomial corresponding to the alliance group of $n$ subject is $a_1 a_2 ...a_n$.
Next we construct the diagonal form and apply folding procedure:

\[\begin{array}{*{20}{c}}
   {} & {[{a_1}][{a_2}]...[{a_n}]} & {}  \\
   {[{a_1}{a_2}...{a_n}]} & {} & { = [{a_1}{a_2}...{a_n}] + \overline {[{a_1}][{a_2}]...[{a_n}]}  = 1 \ .}  \\
\end{array}\]

Therefore the alliance groups are always super-active.

\textit{Conflict groups}.
The polynomial corresponding to the conflict group of $n$ subject is $a_1+ a_2 + ... + a_n$.
Next we construct the diagonal form and apply folding procedure:
\[\begin{array}{*{20}{c}}
   {} & {[{a_1}] + [{a_2}] + ... + [{a_n}]} & {}  \\
   {[{a_1} + {a_2} + ... + {a_n}]} & {} & {=} \\
   { [{a_1} + {a_2} + ... + {a_n}]   + } & {\overline {[{a_1}] + [{a_2}] + ... + [{a_n}]}  = 1 \ .} & {} 
\end{array}\]

Therefore the conflict groups are always super-active. 

Since both the alliance and the conflict groups are super-active, this lemma is proved. $\Box$
\end{proof}

However, there are non-homogenous super-active groups as well (see Appendix \ref{appen3}).

Summarizing this section, we note that subjects in the super-active groups cannot be controlled in their choices and the entire groups is uncontrolable. Therefore, once the super-active groups emerges, the only way to make it controllable is to change the relationships in the group.

%!!!!!!!!!!!!!!!!!!!!!!!!!!!!!!!!!!
\section{The Basic Control Schema of an Abstract Subject (BCSAS)  in the RGT}
\label{bcs}
We have presented the detailed description of the RGT including solution of the Forward and Inverse tasks. We have also considered the extream cases of decisions like putting a subject into frustration state or changing structure of a super-active group. As a final stroke, we summarize all the presented material in the form of \textit{Basic Control Schema of an Abstract Subject (BCSAS) in the RGT}. %The BCS is presented in Fig. \ref{fig:fig311}.
\begin{figure}
\centering
\includegraphics[height=8cm]{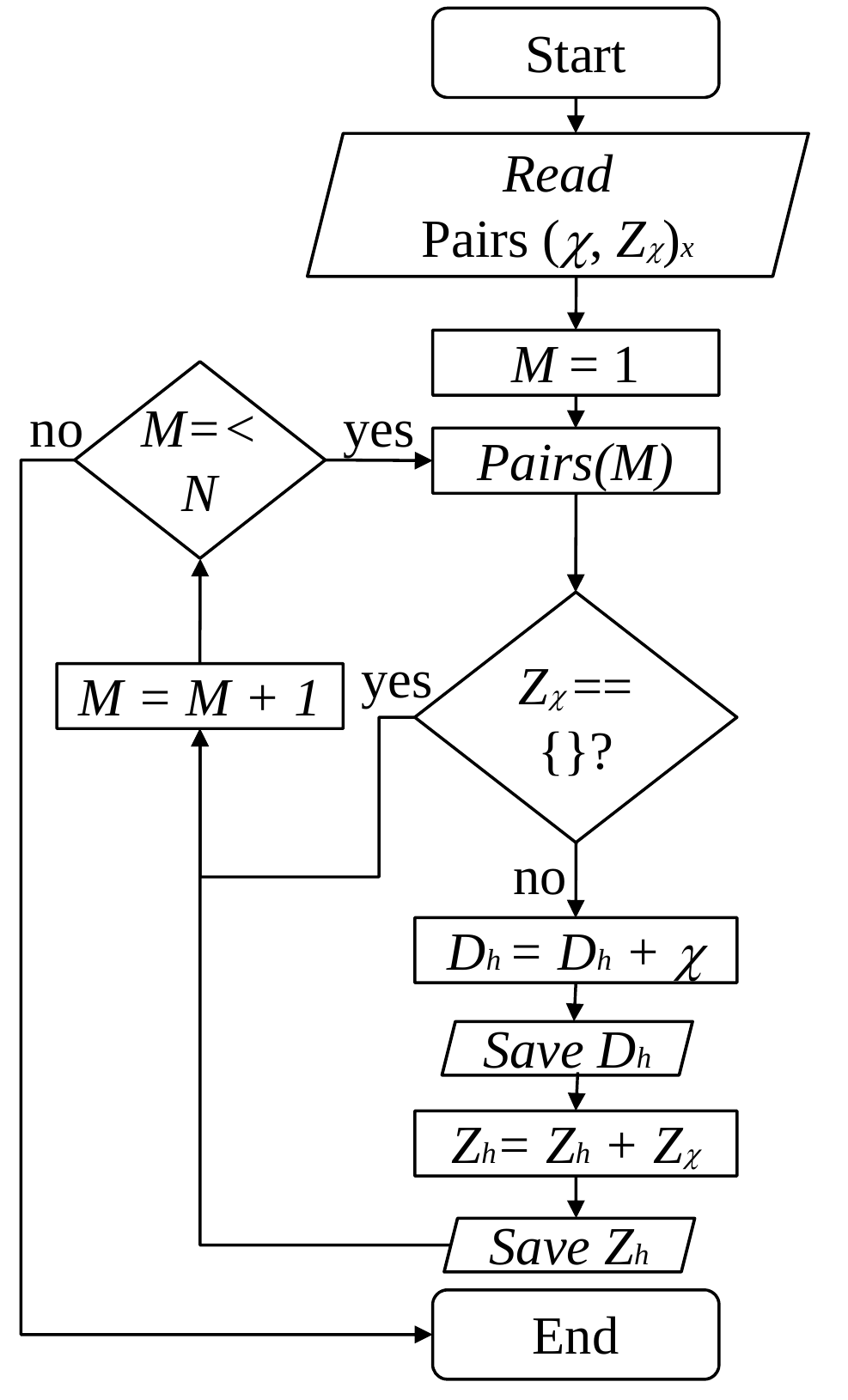}
\caption{The Block schema for extracting sets $D_h$ and $Z_h$.}
\label{fig4}
\end{figure}

The input comes from the environment and is formalized in the form of external Influences on the subject, the Boolean algebra of Alternatives and Structure of a Group.

Information about the Influences, Boolean algebra and Group Structure is propagated into the \textit{Decision Module}.  
The Decision Module implements solution of the Forward task.  Therefore the output set $\mathbb D$  of the Decision Module is the set of possible alternatives, which subject can choose under the given conditions. 

The information about Boolean algebra and Group Structure is propagated into the \textit{Influence Module}. The Influence Module solves the Inverse task. The output set $\mathbb D_h$ of the Influence Module is the set of the pairs $(\chi, \mathcal Z_{\chi})_x$, where $\chi$ is the target alternative, the set $\mathcal Z_{\chi}$ is the set of all the joint influences, resulting in selection of the target choice; and $x$ represents a subject variable. Each $(\chi, \mathcal Z_{\chi})_x$ represents a reflexive control strategy.

Therefore, the decision to put a subject into $frustration$ state is justified if it is impossible to make subject $x$ choose the target alternative $\chi$, i.e., if for pair $(\chi, \mathcal Z_{\chi})_x$ set $\mathcal Z_{\chi} = \{\}$, and subject $x$ should not choose any other alternative except for the target one.

\subsection{Schema for Iterative Algorithm to Obtain Output of the Influence Module}
\label{iterat}

The alternatives $\chi$ with corresponding non-empty sets $\mathcal Z_{\chi}$ are included into the set $\mathbb D_h$. Here we introduce set $\mathbb Z_h$ to store the non-empty sets $\mathcal Z_{\chi}$. 
The schema of the algorithm for extracting sets $\mathbb D_h$ and $\mathbb Z_h$ is 
presented in Fig. \ref{fig4}.  First the sets $\mathbb D_h$  and $\mathbb Z_h$ are empty: $\mathbb D_h = \{\}$ and $\mathbb Z_h = \{\}$. The algorithm reads the set of pairs $(\chi,\mathcal Z_{\chi})_x$ and stores it in array $Pairs(M)$, where $M$ is a counting variable, $N$ is the total number of pairs. Then it is checked for each pairs from array $Pairs$ whether set $\mathcal Z_{\chi}$ is empty: $\mathcal Z_{\chi} == \{\}?$ . If 'yes', the algorithm increments counting variable $M (M = M+1)$ and proceeds to the next pair from array Pairs. If 'no', then alternative $\chi$ is included into the set $\mathbb D_h$($\mathbb D_h = \mathbb D_h + \chi$), set $\mathbb D_h$ is saved, the set  $\mathcal Z_{\chi}$ is included into set $\mathbb Z_h$ ($\mathbb Z_h = \mathbb Z_h +  \mathcal Z_{\chi}$) and set $\mathbb Z_h$ is saved. The process is run while $M \leq N$.

In this iterative algorithm, we separately store the alternatives $\chi$ , which can be chosen by a certian subject, in the set $\mathbb D_h$ and the joint influences $\mathcal Z_{\chi}$ , which result in selection of alternative $\chi$, in the set $\mathbb Z_h$.

Therefore, we should modify the schema of Influence Module in BCSAS as follows. We present elaborated schema, where sub-module "Solution: $\mathbb D_h$" is accompanied with sub-module "Solution: $\mathbb Z_h$". Together these sub-modules are included into the "Solutions" sub-module.

BCSAS is the fundamental schema of an abstract subject, which is used through out the RGT. The BCSAS is presented in Fig.\ref{fig311}.

This concludes the overview of RGT and description of tasks within the scope of the general theory. Therefore, we continue with application of the RGT to the mixed groups of humans and robots.

\begin{figure}
\centering
\includegraphics[height=4cm]{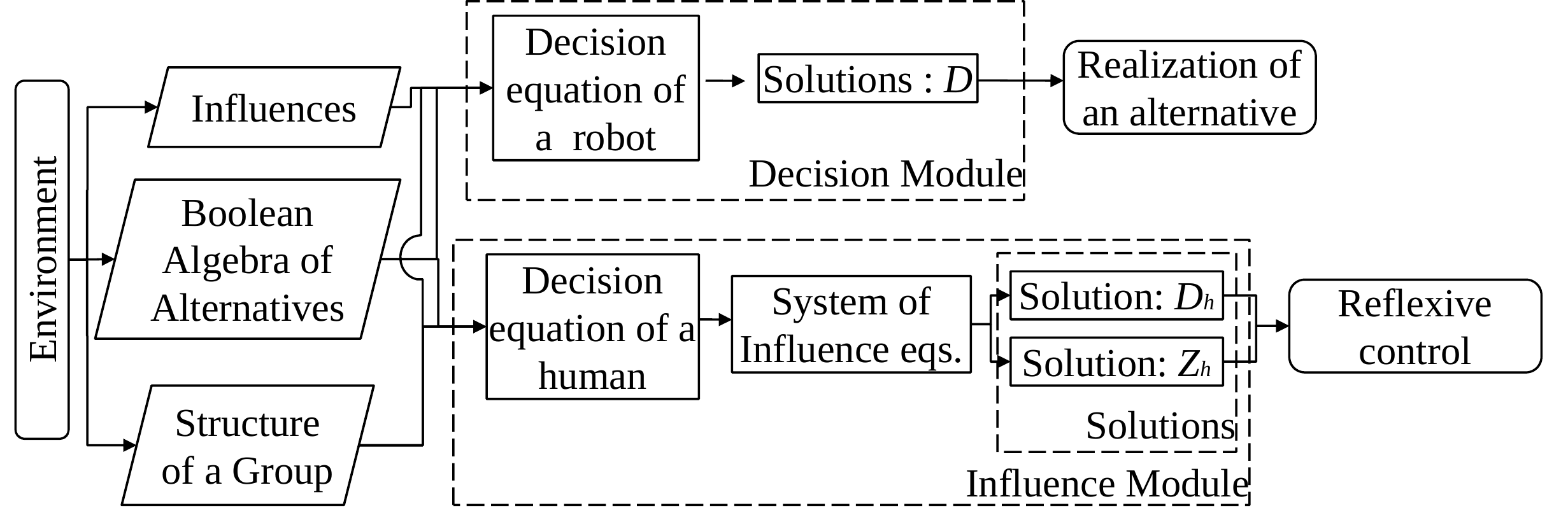}
\caption{The Basic Control Schema of an Abstract Subject (BSCAS).}
\label{fig311}
\end{figure}

\section{Defining Robots in RGT}
\label{robot}
As we have noted in the Introduction section, the goal of the robots in mixed groups of humans and robots is to refrain human subject from choosing risky actions, which might result in injuries or even threaten live.

It is considered by default that robot follows the program of behavior. 
Such program consists of at least three modules. 
The Module 1 implements robot's ability of human-like decision-making based on the RGT. 
The Module 2 contains the rules, which refrain robot from making a harm to human beings. 
The Module 3 predicts the choice of each human subject and suggests the possible reflexive control strategies. 

The Modules 1 and 3 are inhereted from the BCSAS of an Abstract Individual. They correspond to Decision Module and Influence Module of the BCSAS (Fig. \ref{fig311}), respectively. Therefore all the properties and meaning of outputs of the Modules 1 and 3 are the same as the ones for Decision and Influence modules, respectively.

The Module 2 is the new module, which is intrinsic for robotic agents studied in the context of mixed groups of humans and robots. This module is responsible for extraction of only harmless or non-risky alternatives for human subject.

We suggest to apply Asimov's Three Laws of robotics \cite{asi}, which formulate the basics of the Module 2:\\
1) a robot may not injure a human being or, through inaction, allow a human being to come to harm; \\
2) a robot must obey any orders given to it by human beings, except where such orders would conflict with the First Law;\\
3) a robot must protect its own existence as long as such protection does not conflict with the First or Second Law.

We consider that these laws are intrinsic part of robot’s "mind", which cannot be erased or corrupted by any means. 

The interaction of Modules 1 and 2 is performed in the Interaction Module 1. The interaction of Modules 3 and 2 is implements in the Interaction Module 2.

The Boolean algebra is filtered according to Asimov's laws in Module 2. The output of Module 2 is set $\mathbb U$ of approved alternatives.  This data is then propagated into interaction modules.
\begin{figure}
%\centering
\includegraphics[height=9cm]{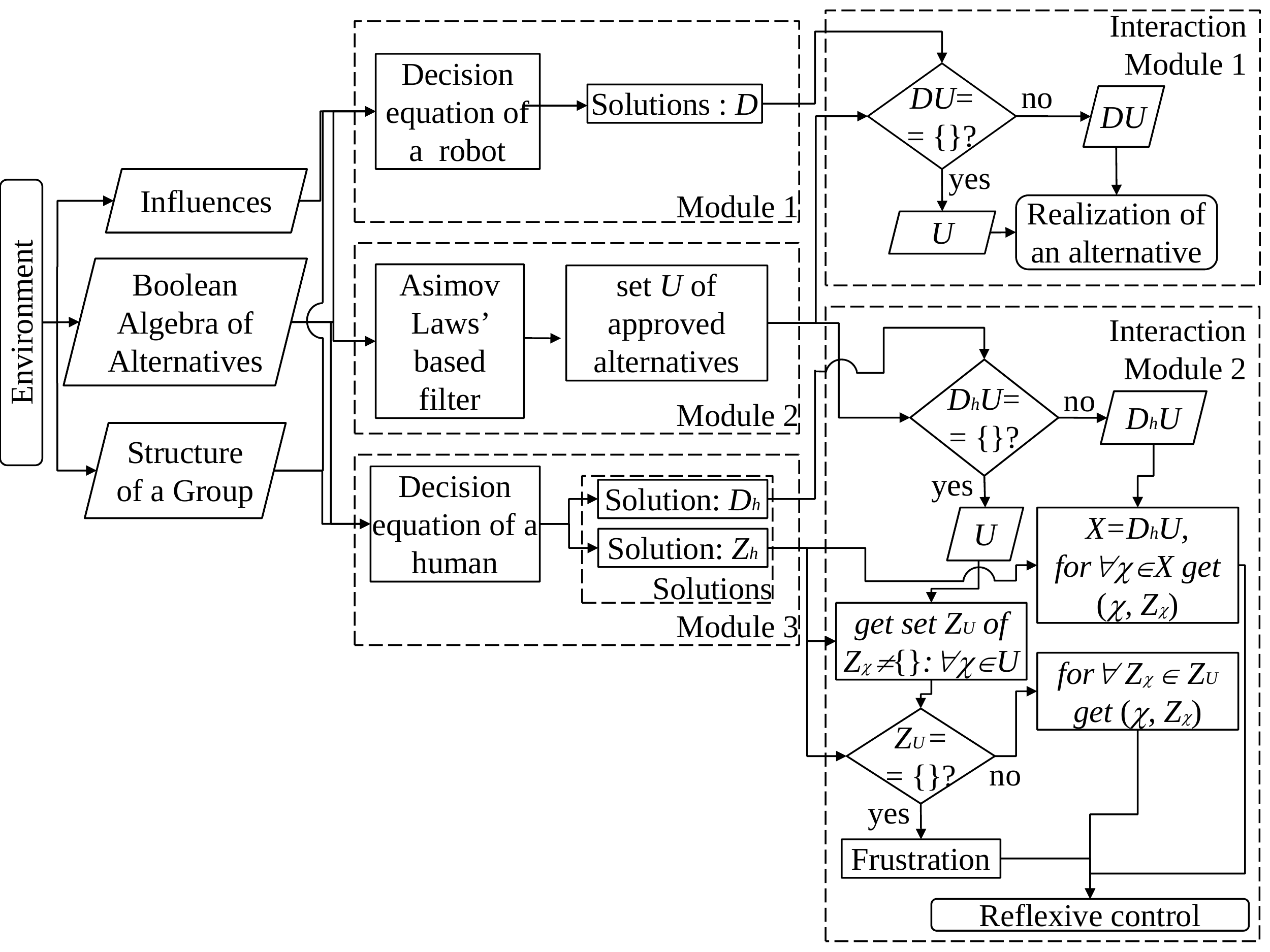}
\caption{The Basic Control Schema of a Robotic Agent (BCSRA).}
\label{fig3}
\end{figure}

The output of the Module 1 is set $\mathbb D$ of alternatives, which robot has to choose under the given joint influences. In the Interaction Module 1, the conjunction of sets $\mathbb D$ and $\mathbb U$ is performed: $\mathbb D \cap \mathbb U = \mathbb{DU}$. If set $\mathbb{DU}$ is not empty set, this means that there are aproved alternatives among the alternatives that robot should choose in accordance with the joint influences. Therefore, robot can implement any alternative from the set $\mathbb{DU}$. If set $\mathbb{DU}$ is empty, this means that under given joint influences robot cannot choose any approved alternative, therefore robot will choose an alternative from set $\mathbb U$. This is how the Interaction Module 1 works.

The output of the Module 3 contains sets $\mathbb D_h$ and $\mathbb Z_h$. The goal of the robot is to refrain human subjects from choosing risky alternative. This can be done by convincing human subjects to choose alternatives from the set $\mathbb U$. First, we check whether $\mathbb D_h$ contains any approved alternative. We do so by performing conjunction of sets $\mathbb D_h$ and $\mathbb U$: $\mathbb D_h \cap \mathbb U = \mathbb D_h \mathbb U$. 

If set $\mathbb D_h \mathbb U$ is not empty, then it means that it is possible to make a human subject to choose some non-risky alternative. Therefore, we should choose the corresponding reflexive control strategy from the set $\mathbb Z_h$. However, if set  $\mathbb D_h \mathbb U$ is empty, we have to find the reflexive control strategy that will make human subject to select approved alternative from set $\mathbb U$. For this purpose, we construct set $\mathbb Z_U$ by including all the joint influences $\mathcal Z_{\chi}$ for approved alternatives: $\mathcal Z_{\chi} \in \mathbb Z_U \Leftrightarrow \chi \in \mathbb U$. Next we check whether set $\mathbb Z_U$ is empty. If set $\mathbb Z_U$ is empty this means it is impossible to convince a human subject to choose non-risky alternative. Therefore, the only option of reflexive control in this case is to put this subject into frustration state. However, if set $\mathbb Z_U$ is not empty, this means that there exist at least one reflexive control strategy that results in selection of alternative from the set of the approved (non-risky) ones.

Therefore, the BCSRA inherits the entire structure of the BCSAS and augments it with Module 2 of Asimov's Laws together with Interaction Modules 1 and 2.

The original schema of robot's control system has been recently presented in \cite{taras}. The BCSRA is extended version of the original schema. The BCSRA provides comprehensive approach of how Forward and Inverse tasks are solved in the robot's "mind".

Thus, in this section we have presented the formalization of robotic agent in the RGT. We outlined the specific features of robotic agents, which distinguish them from other subjects. Furthermore, we provided detailed explanation of how the Forward and Inverse tasks are solved in the framrework of control system (BCSRA) of robots.

Next, we proceed with consideration of sample sutiations of interactions between humans and robots.

%!!!!!!!!!!!!!!!!!!!!!!!!!!!!!!!!!!!!!!!!!!!!!11
\section{Extended Sample Analysis of Mixed Groups}
\label{analysis}
Here we elaborate two examples, presented in the previous study \cite{taras}, of how robots in the mixed groups can make humans refrain from risky actions. We discuss the application of the extended schema of robot's control system and provide explicit derivation of reflexive control strategies, which has been applied in these examples in the prevous study \cite{taras}.

\subsection{Robots Baby-Sitters}
\label{rbb}
Suppose robots have to play a part of baby-sitters by looking after the kids. We consider a mixed group of two kids and two robots. Each robot is looking after a particular kid.
Having finished the game, kids are considering what to do next. They choose between ``to compete climbing the high tree'' (action $\alpha$) and ``to play with a ball'' (action $\beta$). Together actions $\alpha$ and $\beta$ represent the active state 1=$\{\alpha,  \beta\} = \{\alpha\}+\{\beta\}$.
Therefore the Boolean algebra of alternatives consists of four elements: 1) the alternative $\{\alpha\}$ is to climb the tree; 2) the alternative $\{\beta\}$ is to play with a ball; 3) the alternative $1=\{\alpha, \beta\}$ means that a kid is hesitating what to do; and 4) the alternative $0 = \{\}$ means to take a rest.

We consider that each kid considers his robot as ally and another kid and his robot as the competitors. The kids are subjects $a$ and $c$, while robots are subjects $b$ and $d$. The relationship graph is presented in Fig.~\ref{fig:fig4}.     
\begin{figure}
\centering
\includegraphics[height=2cm]{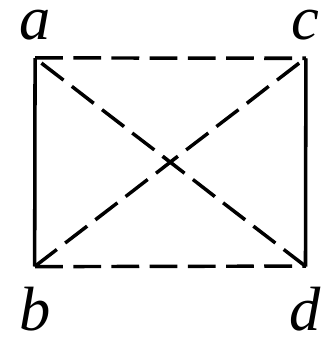}
\caption{The relationship graph for robots baby-sitters examples.}
\label{fig:fig4}
\end{figure}

Next we calculate the diagonal form and fold it in order to obtain decision equation for each subject:
\[\begin{array}{*{20}{c}}
   {} & {} & {[a][b]} & {} & {[c][d]} & {}  \\
   {} & {[ab]} & {} & { + [cd]} & {} & {}  \\
   {[ab + cd]} & {} & {} & {} & {} & { = ab + cd \; .}  \\
\end{array}\]

From two actions $\alpha$ and $\beta$, action $\alpha$ is a risky action, since a kid can fall from the tree and this is real threat for his health or even life. Therefore according to Asimov's laws, robots cannot allow kids to start the competition. Thus, robots have to convince kids not to choose alternative $\{\alpha\}$. In terms of alternatives, the Asimov's laws serve like filters which filter out the risky alternatives. The remaining alternatives are included into set $\mathbb U$. In this case, $\mathbb U = \{\{\beta\},\{\}\}$.

Next we solve the Inverse taks, regarding alternatives $\{\beta\}$ and $\{\}$. 
We conduct the analysis regarding kid $a$. This analysis can be further extended for kid $c$ in the similar manner.

\textit{Solution of the Inverse task for kid $a$ with approved alternatives as target choice.} The decision equation for kid $a$ is $a = ab+cd$. First, we transform it into canonical form: $a = (b+cd)a + cd\overline{a}$.

Next we consider system of influence equations:
\begin{empheq}[left=\empheqlbrace]{align}
 b+cd = \chi \label{sysBS1a} \\
cd = \chi,  \label{sysBS1b} 
\end{empheq}

where alternative $\chi \in \mathbb U$.

Regarding eq.(\ref{sysBS1b}), eq.(\ref{sysBS1a}) is transformed into equation 
\begin{equation}
\label{rbeq1}
b + \chi = \chi
\end{equation}

The solution of eq.(\ref{rbeq1}) directly follows from Lemma \ref{lemma1}: $\chi \supseteq b \supseteq 0$.
Therefore for $\chi = \{\beta\}$ and $\chi = \{\}$ the solutions are $ \{\beta\} \supseteq b \supseteq 0$ and  $b = 0$, respectively. 

The eq.(\ref{sysBS1b}) can be instantly solved according to Lemma \ref{lemma2}: $\chi d + \overline{\chi} \ \overline{d} \supseteq c \supseteq \chi$.

\textit{Consider $\chi = \{\beta\}$ first}. Then $\{\beta\}d  + \{\alpha\} \overline{d} \supseteq c \supseteq \{\beta\}$. By varying values of variable $d$, we obtain all the pairs $(c,d)$:

d = 1: $ \{\beta\} \supseteq c \supseteq \{\beta\} \Rightarrow c = \{\beta\}$. Therefore the solution is pair $(\{\beta\},1)$;

d = 0: $ \{\alpha\} \supseteq c \supseteq \{\beta\}$. Since $\{\alpha\} \cap \{\beta\} = \{\}$, there is no solution;

d =$\{\alpha\}$ : $0 \supseteq c \supseteq \{\beta\}$. Since $\{\beta\} \supseteq \{\}$, there is no solution;

d =$\{\beta\}$ : $1 \supseteq c \supseteq \{\beta\}$. Therefore there are two solutions $(1,\{\beta\})$ and $(\{\beta\},\{\beta\})$.

Therefore equation $cd = \{\beta\}$ has three solutions $(\{\beta\},1)$, $(1,\{\beta\})$ and $(\{\beta\},\{\beta\})$.

Thus, we have solved both equations from system (\ref{sysBS1a}-\ref{sysBS1b}). The solutions of this system are the triplets $(b,c,d)$ of joint influences, which are all possible combinations of solutions of both equations. Since there are two solution of eq.(\ref{sysBS1a}) and three solutions of eq.(\ref{sysBS1b}), there are six triplets $(b,c,d)$ in total: $(0, \{\beta\},1)$ and $(\{\beta\}, \{\beta\},1)$; $(0, 1, \{\beta\})$ and $(\{\beta\}, 1, \{\beta\})$; $(0, \{\beta\}, \{\beta\})$ and $(\{\beta\}, \{\beta\}, \{\beta\})$.

\textit{Now we consider the case, when $\chi = 0 =\{\}$}. Then $\overline{d} \supseteq c \supseteq 0$.
We obtain pairs $(c,d)$ for all values of variable $d$:

$d = 1$: $\overline{1} \supseteq c \supseteq 0 \Rightarrow c = 0$. Thus, there is only one solution (0,1);

$d = 0$: $1 \supseteq c \supseteq 0$. Thus, there are four solutions $(1,0), (\{\alpha\},0), (\{\beta\},0)$ and $(1,0)$;

$d = \{\alpha\}$: $\{\beta\} \supseteq c \supseteq 0$. Thus, there are four solutions $(\{\beta\},\{\alpha\})$ and $(0,\{\alpha\})$;

$d = \{\beta\}$: $\{\alpha\} \supseteq d \supseteq 0$. Thus, there are four solutions $(\{\alpha\},\{\beta\})$ and $(0,\{\beta\})$.

In total, equation $cd = 0$ has 9 solutions. Therefore system (\ref{sysBS3a}-\ref{sysBS3b}) also has 9 solutions as triplets $(b,c,d)$: $(0,1,0)$, $(0,0,0)$, $(0,0,\{\alpha\})$, $(0,0,\{\beta\})$, $(0,0,1)$, $(0,\{\alpha\},\{\beta\})$,
$(0,\{\alpha\},0)$, $(0,\{\beta\},\{\alpha\})$ and $(0,\{\beta\},0)$.

We have considered two cases, when both upper and lower bounds of the interval of decision equation equal to the same alternative. Now we discuss a new situation, when variable $a$ should take not a single value, but several values. In this case, we should find the joint influences $(b,c,d)$ that result in selection of either alternative $\{\beta\}$ or $\{\}$. Since, $\{\beta\} \supseteq \{\}$, we need to find all the triplets $(b,c,d)$, resulting in the solution of decision equation as interval 
$\{\beta\} \supseteq a \supseteq \{\}$. Thus, $\{\beta\} \supseteq a^* \supseteq \{\}$. 

Therefore, we need to solve the following system of equations:
\begin{empheq}[left=\empheqlbrace]{align}
b+cd = \{\beta\} \label{sysBS3a} \\
cd = 0.  \label{sysBS3b} 
\end{empheq}

The eq.(\ref{sysBS3a}) turns into equality $b = \{\beta\}$, and we need to solve eq.(\ref{sysBS3b}). However, this equation has been already solved in the previous example. Therefore we obtian the solutions of the system (\ref{sysBS3a}-\ref{sysBS3b}): $(\{\beta\},1,0),(\{\beta\},0,0)$, $(\{\beta\},0,\{\alpha\}), \\ (\{\beta\},0,\{\beta\}),(\{\beta\},0,1), (\{\beta\},\{\alpha\},\{\beta\}), (\{\beta\},\{\alpha\},0)$,  $(\{\beta\},\{\beta\},\{\alpha\})$ and \\ $(\{\beta\},\{\beta\},0)$.

Comparing solutions of all three system of influence equation, we can see that there are four remarkable solutions  $(\{\beta\},\{\beta\},\{\beta\})$ and $(\{\beta\},\{\},\{\beta\})$; $(\{\beta\},1,\{\beta\})$ and $(\{\beta\},\{\alpha\},\{\beta\})$. The first pair of solution results in choice of only alternative $\{\beta\}$, while second pair of solutions results in selection of eighter alternative $\{\beta\}$ or alternative $\{\}$. These four solutions together illustrate that if $b = d = \{\beta\}$, it is guaranteed that regardless of influence of kid $c$, kid $a$ will choose either of approved alternatives.

By analogy, we can see that among solutions of system (\ref{sysBS1a}-\ref{sysBS1b}) with $\chi = \{\}$, there are four solutions $(0,1,0)$,$(0,0,0)$, $(0,\{\alpha\},0)$ and $(0,\{\beta\},0)$. Therefore, if $b = d = 0$, kid $a$ will choose alternative $0=\{\}$ regardless of influence of kid $c$.

These two examples of binding variables $b$ and $d$ were considered in \textit{Scenario 1} and \textit{Scenario 2} of sample situation with robot baby-sitters, originally presented in \cite{taras}. 

Summarizing the results of this section, we have shown that robots can successfully control kids' behavior by refraining them from doing risky actions. The basic of this control is entirely based on the proposed schema of robot's control system. We have analyzed all the possible reflexive control strategies by solving three systems of influence equation: two systems regarding a single alternative and one system regarding the interval of alternatives. Therefore, we have shown how the Inverse task can be effectively solved by our proposed algorithm in situation similar to the real conditions. 

\subsection{Mountain-Climbers and Rescue Robot} 
\label{mcrr}
We consider that there are two climbers in the mountain and rescue robot. The climbers and robot are communicating via radio. One of the climbers (subject $b$) got into difficult situation and needs help. Suggest, he fell into the rift because the edge of the rift was covered with ice. The rift is not too deep and there is a thick layer of snow on the bottom, therefore climber is not hurt, but he cannot get out of the rift himself. The second climber (subject $a$) wants to rescue his friend himself (action $\alpha$), which is risky action. The second option is that robot will perform rescue mission (action $\beta$). Since inaction is inappropriate solution according to the First Law, the set $\mathbb U$ of approved alternatives for robot includes only alternative $\{\beta\}$. The goal of the robot is to refrain the climber $a$ from choosing alernative $\{\alpha\}$ and perform rescue mission itself.

We suggest that from the beginning all subjects are in alliance. The corresponding graph is presented in Fig.~\ref{fig:fig1}c and its polynomial is $abc$. Therefore by definition it is homogenous group and, consequently, it is super-active group according to Lemma \ref{lemma3}.

%Next we calculate diagonal form and perform folding procedure:
%\[\begin{array}{*{20}{c}}
%  {} & {[a][b][c]} & {}  \\
%  {[abc]} & {} & { = [abc] + \overline {[a][b][c]} }  \\
%   {} & {} & {}  \\
%\end{array} = 1 \; .\]

Thus, any subject in the group is in active state. Therefore, group is uncontrollable (see Section \ref{supac}). In this case, robot makes decision to change his relationship with the climber $b$ from alliance to conflict. Robot can do that, for instance, by not responding to climber's orders. 

\textit{Which reflexive control leads to frustration state?} Then the polynomial corresponding to the new group is $a(b+c)$.  This polynomial has been already broadly discussed in the Section \ref{frust}. Therefore, we know decision equation for subject $a$: $a = (b+c)a + \overline{a}$. We have shown as well that subject $a$ can choose only alternative $1 = \{\alpha,\beta\}$, if appropriate joint influences are applied (see Section \ref{frust}), overwise subject $a$ is in frustration state and cannot make any choice. Therefore, in order to put subject $a$ into frustration state, the reflexive control strategy should $NOT$ be selected from the list of solutions (Section \ref{frust}): $(\{\beta\},\{\alpha\})$; $(1,\{\alpha\})$;  $(\{\alpha\},\{\beta\})$; $(1,\{\beta\})$;  $(0,1)$; $(\{\alpha\},1)$; $(\{\beta\},1)$; $(1,1)$ and $(1,0)$.

Here we provide two examples of such joint influences $(b,c)$: $(\{\alpha\}, \{\alpha\}) \Rightarrow (\{\alpha\}+\{\alpha\}) = \{\alpha\} \subset 1$ and $(\{\beta\}, \{\}) \Rightarrow (\{\beta\}+\{\}) = \{\beta\} \subset 1$.

\textit{Whether robot can complete mission regardless of joint influences of other subjects?}
The decision equation for robot $c$ is $c = c + (b+\overline{a})\overline{c}$. The corresponding solution interval is 
$1\supseteq c \supseteq (b +\overline{a})$.

Here we analyze all 16 possible reflexive control strategies $(a,b)$ that climbers can apply to robot $c$.
\subsubsection{Examples with emtpy set $\mathbb D \mathbb U$.}
For $(0,b)$, there will be the same situation regardless of value of variable $b$ : $1\supseteq c \supseteq (b +\overline{0}) \Rightarrow 1\supseteq c \supseteq (b +1) \Rightarrow c = 1$. 

For  $(a,1)$, there will be the same situation regardless of value of variable $a$ : $1\supseteq c \supseteq (1+\overline{a}) \Rightarrow c = 1$.

For $(\{\alpha\},\{\alpha\})$: $1\supseteq c \supseteq (\{\alpha\} +\overline{\{\alpha\}}) \Rightarrow 1\supseteq c \supseteq (\{\alpha\} +\{\beta\}) \Rightarrow  c = 1$.

For $(\{\beta\},\{\beta\})$: $1\supseteq c \supseteq (\{\beta\} +\overline{\{\beta\}}) \Rightarrow 1\supseteq c \supseteq (\{\beta\} +\{\alpha\}) \Rightarrow c = 1$.
Therefore in these cases  set $\mathbb D = \{\{\alpha,\beta\}\}$. 

Next we consider other pairs $(a,b)$.

$(1,\{\alpha\})$: $1\supseteq c \supseteq (\{\alpha\} +\overline{1}) \Rightarrow 1\supseteq c \supseteq \{\alpha\}$.
Here  set $\mathbb D = \{\{\alpha,\beta\}, \{\alpha\}\}$. 

$(\{\beta\},\{\alpha\})$: $1\supseteq c \supseteq (\{\alpha\} +\overline{\{\beta\}}) \Rightarrow 1\supseteq c \supseteq \{\alpha\}$. Here  set $\mathbb D = \{\{\alpha,\beta\}, \{\alpha\}\}$. 

$(\{\beta\},0)$: $1\supseteq c \supseteq (0 +\overline{\{\beta\}}) \Rightarrow 1\supseteq c \supseteq \{\alpha\}$. Therefore,  set $\mathbb D = \{\{\alpha,\beta\}, \{\alpha\}\}$.

Since $\mathbb U = \{\{\beta\}\}$,  $\mathbb D \mathbb U = \{\}$ for all the cases considered above, robot will choose alternative $\{\beta\}$ from the set $\mathbb U$.

\subsubsection{ Examples with non-empty set $\mathbb D \mathbb U$.}
Consider the following pairs $(a,b)$:

$(1,\{\beta\})$: $1\supseteq c \supseteq (\{\beta\} +\overline{1}) \Rightarrow 1\supseteq c \supseteq \{\beta\}$.
Therefore,  set $\mathbb D = \{\{\alpha,\beta\}, \{\beta\}\}$. 

$(1,0)$: $1\supseteq c \supseteq (0 +\overline{1}) \Rightarrow 1\supseteq c \supseteq 0$.
Thus,  set $\mathbb D = \{\{\alpha,\beta\}, \{\alpha\}, \{\beta\}, \{\}\}$.

$(\{\alpha\},\{\beta\})$: $1\supseteq c \supseteq (\{\beta\} +\overline{\{\alpha\}}) \Rightarrow 1\supseteq c \supseteq \{\beta\}$. Thus,  set $\mathbb D = \{\{\alpha,\beta\}, \{\beta\} \}$.

$(\{\alpha\},\{\beta\})$: $1\supseteq c \supseteq (\{\beta\} +\overline{\{\alpha\}}) \Rightarrow 1\supseteq c \supseteq \{\beta\}$. Thus,  set $\mathbb D = \{\{\alpha,\beta\}, \{\beta\} \}$. 

$(\{\alpha\},0)$: $1\supseteq c \supseteq (0 +\overline{\{\alpha\}}) \Rightarrow 1\supseteq c \supseteq \{\beta\}$. Thus,  set $\mathbb D = \{\{\alpha,\beta\}, \{\beta\} \}$.

Since $\mathbb U = \{\{\beta\}\}$,  $\mathbb D \mathbb U = \{\{\beta\}\}$ for all the cases considered above, robot will choose alternative $\{\beta\}$ from the set $\mathbb D \mathbb U$.\\

Thus, we have shown that under all 16 reflexive control strategies $(a,b)$, robot $c$ can choose the alternative $\{\beta\}$, which is to perform the rescue mission itself. Therefore robot will choose alternative $\{\beta\}$ regardless of the joint influences $(a,b)$ of the climbers.

The discussed example illustrates how robot can transform uncontrollable group into controllable one by manipulating the relationships in the group. In the controllable group by its influence on the human subjects, robot can refrain the climber $a$ from risky action to rescue climber $b$. Robot achieves its goal by putting climber $a$ into frustration state, in which climber $a$ cannot make any decision. On the other hand, set $\mathbb U$ of approved alternatives guarantees that robot itself will choose the option with no risk for humans and implement it regardless of climber's influence.

Therefore, in this section we have illustrated robot's ability to refrain human being from risky actions and to perform these risky actions itself. This proves that our approach achieves both goals of robotic agent: 1) to refrain people from risky actions and 2) to perform risky actions itself regardless of human's influences.

%!!!!!!!!!!!!!!!!!!!!!!!!!! Discussion and Conclusion
\section{Discussion and Conclusion}

Summarizing, the results of this paper, we outline the most important of them. 

First of all, we have introduced the Inverse task and developed the ultimate methods to solve it.

We have provided a comprehensive tutorial to the \textbf{brand new Reflexive Game Theory} recently formulated and proposed by Vladimir Lefebvre \cite{lef1,lef2,lef3,lef4}. The tutoral contains the detailed description of the Forward and Inverse tasks together with methods to solve them. 

We propose control schemas for both abstract subject (BCSAS) and robotic agent (BCSRA). These schemas were specially designed to incorporate solution of the Forward and Inverse tasks, thus providing us with autonomous units (individuals, subjects, agents) capable of making decisions in the human-like manner. We have shown that robotic agents based on BCSRA can be easily included into the mixed groups of humans and robots and effectively serve their fundamental goals (refraining humans from risky actions and, if needed, perform  the risky acions itself).

Therefore, we consider that present study provides the comprehensive overview of the classic RGT proposed by Vladimir Lefebvre \cite{lef1,lef2,lef3,lef4} and newly developed self-consistent framework for analysis of different kinds of groups and societies, including human social groups and mixed groups of humans and robots together with application tutorial of this new framework. 

This framework is entirely based on the principles of the RGT and brings together all its elements. The solution of the Inverse task, presented in this paper, plays a crutial role in formation of this framework.  Therefore, by having the Inverse task as one of its fundamentals, this framework illustrates the role of the Inverse task and its relationship with other issues considered in the RGT.

%\section*{Acknowledgement}


%\section*{References}

\section*{Appendix}
\appendix
\section{When sets $A$ and $B$ are functions of less than total number of subject minus one variables}
\label{appen1}
Consider groups of four subjects $a,b,c$ and $d$. Suggest the polynomial corresponding to this group is $b(a+d)+c$. 
Next we construct diagonal form and perform folding operation:

\[\begin{array}{*{20}{c}}
   {} & {} & {} & {[a] + [d]} & {} & {}  \\
   {} & {} & {[b][a + d]} & {} & {} & {}  \\
   {} & {[b(a + d)]} & {} & {} & { + [c]} & {}  \\
   {[b(a + d) + c]} & {} & {} & {} & {} &  =   \\
\end{array}\]

\[\begin{array}{*{20}{c}}
   {} & {} & {[b]([a + d] + \overline {[a] + [d]} )} & {} & {} \\
   {} & {[b(a + d)]} & {} & { + [c]} & {} \\
   {[b(a + d) + c]} & {} & {} & {} &  =  \\
\end{array}\]

\[\begin{array}{*{20}{c}}
   {} & {} & {[b]} & {} & {} \\
   {} & {[b(a + d)]} & {} & { + [c]} & {} \\
   {[b(a + d) + c]} & {} & {} & {} &  =  \\
\end{array}\]

\[\begin{array}{*{20}{c}}
   { = b(a+d)+c + \overline{b(a+d)+\overline{b}+c} }  \\
\end{array}\]

Next we simplify the resultant expression of diagonal form folding:
\[\begin{array}{*{20}{c}}
   {b(a+d)+c + \overline{b(a+d)+\overline{b}+c} = b(a+d)+c + \overline{b(a+d)}\overline{c}b =}  \\
   {b(a+d)+cb + c\overline{b} + \overline{b(a+d)}\overline{c}b = b((a+d) + c + \overline{b(a+d)}\overline{c}) + c\overline{b} = }  \\
{b((a+d)c + (a+d)\overline{c} + c + (\overline{b} + \overline{(a+d)})\overline{c}) + c\overline{b} = }  \\
{b((a+d)c + (a+d)\overline{c} + c + \overline{b}\overline{c} + \overline{(a+d)}\overline{c}) + c\overline{b} = }  \\
{b(c + (a+d)c + ((a+d)+\overline{(a+d)})\overline{c}) + c\overline{b} = b((a+d)c + c + \overline{c}) + c\overline{b} =}  \\
{b((a+d)c + 1) + c\overline{b} = b + c\overline{b}} = b + c\\
\end{array}\]

Consequently,
\[\begin{array}{*{20}{c}}
   {} & {[b(a + d)] + \overline {[b]}  + [c]} & {}  \\
   {[b(a + d) + c]} & {} &  {=b + c}\\
\end{array}\]

Therefore, the decision equation includes only two subject variables instead of four. Consequenly, for subjects $a$ and $d$ the decision equations in canonical forms are 
\begin{eqnarray}
a = (b + c)a + (b + c)\overline{a} \\
\label{dc11a}
d = (b + c)d + (b + c)\overline{d}
\label{dc12a}
\end{eqnarray}

Thus, the sets $A$ and $B$ for subjects $a$ and $d$ are equal. The sets $A$ and $B$ are functions of only variables $b$ and $c$: $A = A(b,c) =b + c\overline{b} $ and $B = B(b,c) =b + c\overline{b}$.

The canonical forms of decision equations for subjects $b$ and $c$ are:
\begin{eqnarray}
b = b + c\overline{b}\\
\label{dc13a}
c = c + b\overline{c}
\label{dc14a}
\end{eqnarray}

Therefore, set $A=1$ for both subjects. Set B is a functions of a single variable: $B(c) = c$ and $B(b) = b$ for subjects $b$ and $c$, respectively.

\section{Example of non-homogenous super-active groups}
\label{appen3}
Here we provide an example of non-homogenous super-active group.

Consider the group of four subject $a,b,c$ and $d$, which is described by polynomial $c(ab+b)$. Let us build the diagonal form and perform its folding:

\[\begin{array}{*{20}{c}}
   {} & {} & {} & {[a][b]} & {} & {}  \\
   {} & {} & {[ab]} & {} & { + [d]} & {}  \\
   {} & {[c][ab + d]} & {} & {} & {} & {}  \\
   {[c(ab + d)]} & {} & {} & {} & {} &  =   \\
\end{array}\]

\[\begin{array}{*{20}{c}}
   {} & {} & {([ab] + \overline {[a][b]}) + [d]} & {}  \\
   {} & {[c][ab + d]} & {} & {} \\
   { = [c(ab + d)]} & {} & {}  & {=} \\
\end{array}\]

\[\begin{array}{*{20}{c}}
   {} & {} & {1} & {}  \\
   {} & {[c][ab + d]} & {} & {} \\
   { = [c(ab + d)]} & {} & {}  & {=} \\
\end{array}\]

\[\begin{array}{*{20}{c}}
   { = [c(ab + d)] + \overline{[c][ab + d]} = 1 \ \ \Box}
\end{array}\]

\end{document}